\newcommand{\authorrefnametobedeleted}[1]{#1}
\renewcommand{\authorrefnametobedeleted}[1]{}

\documentclass[dvips,bjps]{imsart}

\RequirePackage[OT1]{fontenc}
\usepackage{amsthm,amsmath,amssymb,natbib}
\RequirePackage{hyperref}

\DeclareMathOperator*{\esssup}{ess~sup}
\DeclareMathOperator*{\essinf}{ess~inf}

 \newcommand{\op}[1]{{\mathrm{#1}}}
 \newcommand{\kl}[1]{\ensuremath{\mathrm{#1}}}
\renewcommand{\d}{\hspace{0.1ex}\mathrm{d}\hspace{-0.1ex}}

 \newcommand{\avdot}{\:\cdot\:}
 \newcommand{\E}{\op{E}}                 %
 \renewcommand{\Pr}{\op P}
 \newcommand{\Pj}{\Pr_{\negthinspace j}}

 \newcommand{\sn}[1]{s^{(n)}_{#1}}
 \newcommand{\period}{q}
 \newcommand{\units}{u}
 \newcommand{\JODD}{J}
 \newcommand{\ELL}{\kappa}
 \newcommand{\jT}{j_*}

\arxiv{1206.5756}

\startlocaldefs
\numberwithin{equation}{section}
\theoremstyle{plain}
\newtheorem{thm}{Theorem}[section]
\newtheorem{cor}[thm]{Corollary}

\newtheorem{prop}[thm]{Proposition}
\theoremstyle{definition}
\newtheorem{defn}[thm]{Definition}
\newtheorem{assume}[thm]{Assumption}

\newtheorem{asno}[thm]{Assumption/notation}
\newtheorem{rem}[thm]{Remark}
\theoremstyle{remark}

\newcommand{\slutt}{}
\endlocaldefs

\begin{document}

\begin{frontmatter}
  \title{On free lunches in random walk markets with short-sale
    constraints and small transaction costs, and weak convergence to
    Gaussian continuous-time processes\thanksref{t1}} %
  \runtitle{Free lunches in random walk markets} \thankstext{t1}{This
    paper has benefited from discussions with in particular Albert
    Shiryaev and Gisle James Natvig, and (indeed, more than usual)
    from comments from the journal's referees.}

\begin{aug}
  \author{\fnms{Nils Chr.}
    \snm{Framstad}\thanksref{a}\thanksref{b}\ead[label=e1]{ncf+research@econ.uio.no}},

\runauthor{N.C. Framstad}

\affiliation[a]{University of Oslo, Department of Economics}
\affiliation[b]{The Financial Supervisory Authority of
  Norway\footnote{Content does not reflect the views of The Financial
    Supervisory Authority of Norway.}}

\address{University of Oslo, Department of
    Economics, P.O.\ Box 1095 Blindern, NO-0317 Oslo, Norway.\\\printead{e1}}

\end{aug}

\begin{abstract}
  This paper considers a sequence of discrete-time random walk markets
  with a safe and a single risky investment opportunity, and gives
  conditions for the existence of arbitrages or free lunches with
  vanishing risk, of the form of waiting to buy and selling the next
  period, with no shorting, and furthermore for weak convergence of
  the random walk to a Gaussian continuous-time stochastic
  process. The conditions are given in terms of the kernel
  representation with respect to ordinary Brownian motion and the
  discretisation chosen. Arbitrage and free lunch with vanishing risk
  examples are established where the continuous-time analogue is
  arbitrage-free under small transaction costs -- including for the
  semimartingale modifications of fractional Brownian motion suggested
  in the seminal \citet{MR1434408} article proving arbitrage in fBm
  models.
\end{abstract}



\end{frontmatter}

\section{Introduction}

If a continuous-time model for a financial market is discretised in
time, will then the discretised version inherit its properties when it
comes to free lunches, or absence of such? Asking the converse
question: if a continuous-time model is the weak limit -- ``weak''
because this topology gives neighbouring profits/loss process
distributions for a given strategy -- of a sequence of discrete-time
models, will free lunch properties or no free lunch properties carry
over the limit transition?

There is actually no guarantee that this will be the case. In
Shiryaev's book \citeyearpar[sec.\ VI.3]{MR1695318}, there are given
stronger sufficient conditions for convergence to fair prices in terms
of weak convergence of the \emph{(driving noise, pricing kernel)
  pair}.  This paper will show that if this joint convergence fails,
then there is a wide range of problems where the arbitrage properties
differ between the discretised prices and their weak limits, even when
small transaction costs are introduced to the former.

This author's initial interest in the problem at hand, emerges from a
work by \authorrefnametobedeleted{Sottinen }\citet{MR1849425}, who
establishes a sequence of discrete-time binary symmetric random walk
(semimartingale) markets, which (a) converges weakly to a
Black--Scholes market with prices being geometric fractional Brownian
motion with Hurst parameter $H>1/2$, and (b) admits an arbitrage
obtained by waiting for the right moment to buy (if nonnegative drift)
or short sell (if non-positive drift) the stock, and unwinding the
position the very next period; the ``right moment'' is of course when
you might with probability one know that the stock market beats the
money market even if tomorrow is a bad day (in which case you buy), or
waiting for the conversely adverse stock market (in which case you
short-sell). Now fractional Brownian motion is not a semimartingale,
and as is well known since \authorrefnametobedeleted{Rogers
}\citet{MR1434408} (for the positively autocorrelated parameter
range), it will introduce arbitrages to canonical models where the
ordinary Brownian motion does not.  In view of this, there seems to
have been a view that the result of \authorrefnametobedeleted{Sottinen
}\citet{MR1849425} is due to specifics of the fBm, or at least its
non-semimartingale property, and this author admits to having fallen
prey to this interpretation, which -- as we shall see -- is
inaccurate.

This paper sets out to show that the phenomenon discovered by
\citet{MR1849425}, is to be expected way more generally, including in
the discretisation of arbitrage-free semimartingale price processes.
As an example, we refer to \citet{MR1434408}, who also proposes a
parametrised semimartingale process whose moving average kernel
converges to the fBm's -- preserving the long memory which was the
reason for suggesting fBm as a driving noise in the first place, but
eliminating the short memory which caused the arbitrage.  It turns out
that when attempting to discretise in a manner akin to the
construction of \citet{MR1849425}, the \emph{long\/} memory will
introduce arbitrages, and the arbitrage property is robust enough to
withstand even the introduction of a small transaction cost.  It is
then essential that the discretised version has bounded downside (cf.\
the results of \authorrefnametobedeleted{Guasoni and coauthors in
}\citet{MR2239592}, \citet{MR2398764}). A different example, admitting
free lunch with vanishing risk (FLVR) in the discretisation, is the
Ornstein--Uhlenbeck process.

We shall on one hand give sufficient conditions for the existence of
arbitrage or FLVR of the form (i) wait for a possible time to buy, and
then (ii) sell next period. On the other, we give sufficient
conditions for weak convergence of the discrete random walks to the
Gaussian continuous-time counterparts.  Because the results concerning
arbitrages will require \emph{bounded\/} innovations in the random
walks, the weak convergence result (Theorem~\ref{weakconv}) will also
be restricted to this case. Our main contributions compared to the
previous literature (primarily \citet{MR1849425}), are summarised as
follows:
\begin{itemize}
\item We cover a fairly general class of Gaussian processes, and give
  examples to the existence of arbitrages/FLVRs of the above-mentioned
  form.
\item Furthermore, we point out that the arbitrage for discretised fBm
  can emerge from the (originally desirable) long-run memory of the
  process, even if the short-run memory (which causes the arbitrage in
  the continuous-time model) is modified as to obtain the
  semimartingale property, e.g.\ as suggested by
  \authorrefnametobedeleted{Rogers }\citet{MR1434408}.
\item We cover any negative drift term (a word which should be
  interpreted cautiously for non-semimartingales) without shorting, as
  it turns out that the instantaneous growth from the noise term can
  tend to infinity.
\item For the same reason the arbitrage may also admit sufficiently
  small \emph{transaction costs}.
\item We do not have to assume the discretised market to be binary
  (hence complete if arbitrage-free) with symmetric innovations. We
  will however assume bounded support, where the bound might depend on
  how fine the discretisation.
\item Weak convergence of the driving noise is likewise shown in this
  more general setting.
\end{itemize}

\section{The continuous-time and discrete-time market
  models} \label{sec:Market} %
Our market has one ``safe'' asset, taken as num\'eraire and normalised
to price $=1$, and one ``risky'' asset $S^{(n)}$, which for each $n$
is a discretisation of a continuously evolving stochastic process $S$.
$S$ will be constructed from a drift process $A$ with time-derivative
$a(t)$ and a driving noise $Z$, assumed to be a Gaussian moving
average process with right-continuous sample paths and an adapted
(hence upper limit of integration is $t$) kernel representation
\begin{equation}\label{Z}
\begin{split}
    Z(t)%
    &=\int_{-\infty}^{t_0} K(t,s) \d W(s) + \int_{t_0}^t
K(t,s) \d W(s) \\&= \JODD(t) + \int_{t_0}^t K(t,s) \d W(s),\end{split}
\end{equation}
with respect to standard Brownian motion $W$, where $K$ is a given
function satisfying the following properties:
\begin{subequations}
\begin{gather}
\text{$K$ is deterministic and piecewise-continuous,}\label{Ka}\\
K(t,s)=0 \quad\text{ if }\quad s\geq t
\intertext{and}
\int_{-\infty}^t(K(t,s))^2\d s<\infty, \quad\forall t.\label{Kb}
\end{gather}\label K 
\noindent We assume that the agent enters the market at a given time
$t_0\geq0$.
\end{subequations}
\\

Notice that \eqref{Kb} follows from the assumed Gaussian distribution,
as the expression is in fact $\E[(Z(t))^2]$. Notice also that some
representations involve $K$ with a definition split in order to
achieve square integrability, compensating the distant past. A
frequently occurring representation form (cf.\ e.g.\
\citet{MR2024843}) is $K(t,s)=\ELL(t-s) - \ELL(-s)$, and under the
assumption that the kernel vanishes for $s\geq t$, then we have
$\ELL(-s)=0$ for $s\geq0$. We shall later give particular attention to
such a form of the type
\begin{math}K(t,s)=\ELL((t-s)_+)\end{math} for $t>s>t_0$, where we not
need to specify the definition below $t_0$.
\\

We might choose to discretise $W$ on the entire time line; however,
$\JODD$ will merely enter as a drift term, and we can equally well
discretise $\JODD$ directly. We shall choose to do the latter. Hence
we start by discretising the time scale (equidistantly) in intervals
of length $1/n$, where for each $n$ we define
\begin{gather}
  \sn i=t_0+\tfrac {i-\lfloor n t_0\rfloor}n,\label{sn}
\end{gather}
where $\lfloor\avdot\rfloor$ is the floor function (rounding toward
$-\infty$). Then we discretise $W$ for $t>t_0$ by replacing its
normalised increments $n^{1/2}\cdot\big[W(\sn{i+1})-W(\sn i)\big]$ by
random variables $\xi_{i+1}=\xi_{i+1}^{(n)}$.  Now discretise $Z$ into
\begin{align}\label{Zn} 
  Z^{(n)}(t)&=\JODD(\tfrac{\lfloor nt\rfloor}n)
  +\sum_{i=\lfloor n t_0\rfloor}^{\lfloor nt\rfloor -1} 
  K(\tfrac{\lfloor nt\rfloor}n, \sn i)\cdot
  n^{-1/2}\xi^{(n)}_{i+1}
\end{align}
For $A$ we can take $A(t_0)=0$, as we are interested in increments
only; we therefore define $A$ and its discretisation as
\begin{align}
    A(t)%
&=\int_{t_0}^ta(s)\d s, &    A^{(n)}(t)&=\frac1n\sum_{i=\lfloor n t_0\rfloor}^{\lfloor nt\rfloor -1} a(\sn i)\label{A}
\intertext {and finally, $S$ and its discretisation are assumed,
  resp.\ defined, to satisfy} S&={G(}A+Z{)} &
S^{(n)}&={G(}A^{(n)}+Z^{(n)}{)}.\label{G}
\end{align}
The canonical choice is $G$ to be the exponential function, but we
shall not need this specific property; for
Proposition~\ref{simplevsfull} we will however use convexity, and for
Theorem~\ref{weakconv} we shall need continuity. Except when $K$
vanishes, the $S^{(n)}$ and the $\xi_i$ sequence will generate the
same filtration, so the first of the following assumptions is not very
restrictive:
\begin{asno} We assume formulae~\eqref{Z} through~\eqref{G} to hold,
  and furthermore:
\begin{itemize}
\item The filtration will be generated by the $\{\xi_i\}$, so that the
  information at time $t$, is generated by $\{\xi_i\}_{i\leq tn}$.
\item By ``step number $j$'', we shall mean at time $\sn j$. That
  means that the agent's first chance of trading, is not at step $0$,
  but at prices noted at step $j_0:=\lfloor nt_0\rfloor$. Should this
  lead to a singularity due to e.g.\ $t_0=0$, $K(t,0)=+\infty$, then
  we shall however eliminate this by assuming (without mention) that
  $t_0$ is $>0$ and irrational.
\item We shall use the term ``$j$-measurable'' to mean measurable at
  \emph{step number\/} $j$, i.e.\ at time $\sn j$, and write
  $\Pj=\Pj^{(n)}$ for the probability measure conditional on the
  filtration generated up to this time/step and $\E_j=\E_j^{(n)}$ for
  the corresponding conditional expectation.
\item The $\{\xi_{i}^{(n)}\}_{i,n}$ will be mutually independent and
  each $\xi_{i}^{(n)}$ bounded, and there exists some (common)
  constant $\nu>0$ such that $\essinf \xi_{i}^{(n)}<-\nu$ and
  $\esssup\xi_{i}^{(n)}>\nu$.
\item Since $\xi_j$ is independent of the past, we shall suppress the
  dependence of law in terms like e.g.\ $\esssup_{\xi_j}$ which will
  denote the supremum over the ($\Pj${}-)essential support of $\xi_j$.
\item Two pieces of notation: $K'_1$ shall denote the partial
  derivative with respect to the first variable. The symbol
  $\eqslantgtr$ shall mean ``no smaller than and not a.s.\ equal''.
\item $a$ is assumed locally bounded, and $G$ is assumed continuous
  and strictly increasing.\slutt
\end{itemize}
\end{asno}

It should be remarked that it is unreasonable for an approximation to
normalised standard Brownian motion that $\nu<1$, but only in parts of
Theorems~\ref{thm:L} and~\ref{thm:FLVR} shall we actually need that
$0$ is interior in the support. Bounded support will however be
essential for the arbitrage conditions, and the following result will
be simplified by assuming a common bound:

\begin{thm}[Weak convergence]\label{weakconv} 
Suppose
  $\E[\xi_i]=0$, $\E[\xi_i^2]=1$ and $\esssup |\xi_i|\leq M<\infty$ (all $i$,
  $n$).
  Then $Z^{(n)}$ converges weakly to $Z$, and for continuous $G$ also
  $S^{(n)}$ to $S$, on the Skorohod space $D([t_0,T])$, every $T>t_0$.
  \end{thm}
\begin{proof}
  The drift and the already occurred part will represent no issue, and
  we can take $A=A^{(n)}=\JODD=0$ without loss of generality. Also, we
  can take $G$ to be the identity, as weak limits commute with
  continuous functions $G$. Now convergence in finite-dimensional
  distributions follows like in \citet[Theorem 1]{MR1849425}: by the
  CLT, the limit is Gaussian with zero mean; for the covariances, the
  independence of the $\xi_i$'s yields, for $T\geq t\geq t_0$
  \begin{align}
    \E[Z^{(n)}(T)Z^{(n)}(t)]&=\sum_{i=\lfloor n t_0\rfloor}^{\lfloor
      nt\rfloor -1} K(\tfrac{\lfloor nT\rfloor}n, \sn i)
    K(\tfrac{\lfloor nt\rfloor}n, \sn i)\cdot
    \E\Big[\Big(\frac{\xi^{(n)}_{i+1}}{\sqrt n}\Big)^2\Big]
  \end{align}
  which is a Riemann sum converging to the desired value $
  \int\limits_{t_0}^tK(T,s)K(t,s)\d s$.

  It remains to prove tightness, which will follow by a set of
  sufficient conditions given in \authorrefnametobedeleted{Whitt
  }\citet[Lemma 3.11 (ii.b)]{MR2368952}. For $T\geq t+h\geq t\geq
  t_0$, we have
  \begin{align}
    \E_{j}[(Z^{(n)}(t+h)&-Z^{(n)}(t))^2]\notag\\&=%
    \sum_{i=\lfloor n t_0 \rfloor}^{\lfloor n(t+h)\rfloor -1}
    \big(K(\tfrac{\lfloor n(t+h)\rfloor}n, \sn i) - K(\tfrac{\lfloor
      nt\rfloor}n, \sn i)\big)^2
    \E_{j}\Big[\Big(\frac{\xi^{(n)}_{i+1}}{\sqrt n}\Big)^2\Big]%
    \notag \\%
    &\leq \frac{M^2}n\sup_{t\in(t_0,T-h)} \sum_{i=\lfloor n t_0
      \rfloor}^{\lfloor n(t+h)\rfloor -1} \big(K(\tfrac{\lfloor
      n(t+h)\rfloor}n, \sn i) - K(\tfrac{\lfloor nt\rfloor}n, \sn
    i)\big)^2 \notag
    \intertext{For each $n$, let $t_n$ be an argsup. For any
      subsequence of $t_n$ converging to a limit point $\bar t$, we
      have convergence as a Riemann sum:}%
    &\rightarrow M^2 %
    \int_{t_0}^{\bar t +h}\big(K(\bar t+h, s) - K(\bar t,s)\big)^2\d s
  \end{align}
which by square integrability tends to $0$ as $h$ does.
\end{proof}

For the discrete-time markets, we shall restrict ourselves to the
following set of strategies:

\begin{defn}
  Let $n<\infty$ be given%
  . For any natural $\period$, a ``$\period$-period strategy''
  (``single period strategy'' if $\period=1$), consists of waiting
  until some stopping time $t_*=\sn {j_*}\geq \sn{j_0}$, buying a
  $j_*$-measurable number $\units>0$ of units, holding these until a
  stopping time $t^*=\sn{j^*}$ where
  $j^*\in\{j_*+1,\dots,j_*+\period\}$ and then selling all $\units$
  units.
  \\
  The ``net return'' from this transaction is
\begin{align}
  R=R_{j_*,j^*}:=\units\cdot(S^{(n)}(t^*)-
  S^{(n)}(t_*))-(\Lambda_*+\Lambda^*)\cdot\lambda
\end{align}
where $\lambda\Lambda_*$ and $\lambda\Lambda^*$ are the respective
transaction costs for buying and selling, allowed to depend on prices
and units like in Assumption~\ref{assume_transcost}
below. \hfill\slutt\end{defn}

The reason for the ``$\lambda$'' parameter is that we will consider
the properties for small transaction costs, and it will be convenient
to scale by a number. The main results will be carried out under for
fixed transaction costs and $\units=1$, and Proposition~\ref{simplevsfull}
will show that this is sufficiently general. For the time being,
assume the more parsimonious form for $\Lambda_*$,
$\Lambda^*$:

\begin{assume}\label{assume_transcost}
  $\Lambda_*=\Lambda_*(\units,S^{(n)}(t_*))$ and
  $\Lambda^*=\Lambda^*(\units,S^{(n)}(t_*), S^{(n)}(t^*))$ will be
  nonnegative functions, bounded in $(S^{(n)}(t_*), S^{(n)}(t^*)$,
  while $\lambda$ will be a number $\geq0$.\hfill\slutt
\end{assume}
\begin{defn}
  We shall use the term ``transaction cost $\lambda$'' to imply that
  $\units=1$ and $\Lambda_*+\Lambda^*=1$ (identically), and we shall
  refer to ``the simple model \eqref{profit}'' the single period case
  of transaction cost $\lambda$ where $G$ is the identity.\hfill\slutt
\end{defn}

\newcommand{\xxx}{x} \newcommand{\yyy}{y} \newcommand{\zzz}{z} %
This ``simple model'' will be the main focus. For this case, the net
return on the event $\{j_*<\infty\}$ will be
\begin{subequations}\label{profit}
\begin{align}
  S^{(n)}(t_0+\tfrac {j_*+1-\lfloor n t_0\rfloor}n)
  -{}&{}S^{(n)}(t_0+\tfrac {j_*-\lfloor n t_0\rfloor}n)-\lambda = %
  \xxx_{j_*}^{(n)}+\yyy_{j_*}^{(n)}+\zzz_{j_*+1}^{(n)}-\lambda\label{profit_a}{,\qquad\text{}}
\end{align}
{where we introduce the notation}
\begin{align}
  \xxx_{j}=\xxx_{j}^{(n)}&=\frac1n \Big\{a(\sn {j})+ \JODD(\sn
  {j+1})-\JODD(\sn j)\Big\}%
  \\
  \yyy_{j}= \yyy_{j}^{(n)} &= \frac1 {\sqrt n}
  \sum_{i=j_0}^{j-1}\Big[K(\sn{j+1},\sn i) - K(\sn j,\sn i)\Big]\:\xi^{(n)}_{i+1} \label{yyy}\\
  \zzz_{j+1}= \zzz_{j+1}^{(n)}&= \frac1 {\sqrt n} \: K(\sn{j+1},\sn j)\:\xi^{(n)}_{j+1}  \label{zzz}
\end{align}
\end{subequations}
adopting the convention that the empty sum, corresponding to $j=j_0$
($=\lfloor nt_0\rfloor$), is zero.  Notice that the $\zzz_{j+1}^{(n)}$
term will represent the \emph{innovation\/} from step $j$ to $j+1$,
and has subscript ``$j+1$'' since it is only $j+1$-measurable. The
$\xxx_j$ and $\yyy_j$, which correspond to the memory of the process
as well as the drift $a$, are $j$-measurable. An arbitrage will occur
in the market if the memory contribution dominates even in the
worst-case innovation. The next section will make this more precise.

\section{Free lunches: sufficient conditions}
\label{sec:Arbitrage}

Starting out with the definitions from the previous section, we now
define arbitrage and free lunches with vanishing risk under our
admissibility conditions. Informally, we have a FLVR if we can obtain
an arbitrarily small downside to mean return ratio, and an arbitrage
if one can have positive-mean return without downside.  The arguably
most natural, and also the strictest, concept of ``downside'' is the
worst-case outcome, the essential supremum of the negative part, and
we shall restrict ourselves to such a definition. The following
definition appears notationally a bit cryptic, but will under
$\period$-period strategies coincide with the conventional definition
of FLVR and arbitrage; informally, it says that downside should be
arbitrarily small compared to expected return (and we note that the
expectation term is finite, by boundedness of the $\xi_i$).  Note
however that this diverges on the outset -- though not in substance,
as we shall see below -- from a commonplace assumption of fixed
horizon. This fixed horizon is less natural here, where the strategies
involve waiting first and only then is there a bound on the holding
period:

\begin{defn}
  Fix $n{, \period \text{ both } <\infty}$. Consider the condition
  \begin{align}
\frac{\essinf^{(\Pr_{j_*})}\; [R_{j_*,j^*}\big|D_{j_*}]}{\op
  E[R_{j_*,j^*}\big|D_{j}]}> -\delta 
\qquad\text{and}\qquad \op E[R_{j_*,j^*}\big|D_{j_*}]\in(0,\infty)
\label{ratio*}
\end{align}
\begin{itemize}
\item The market is said to admit \emph{free lunch with vanishing
    risk\/} (``FLVR'') \emph{in $\period$ periods\/} if for every
  $\delta>0$ there exist stopping times $j_*{\geq j_0}$ and
  $j^*\in\{j_*+1,\dots,j_*+\period\}$ and a $j_*$-measurable event
  $D_{j_*}$ with $\Pr_{j_0}[D_{j_*}]>0$, such that \eqref{ratio*}
  holds.
\item The FLVR is called an \emph{arbitrage\/} if the FLVR definition
  holds also for $\delta=0$.
\item The simple model \eqref{profit} will be said to admit FLVR
  resp.\ arbitrage, if the respective definition applies with
  $\period=1$ (i.e.\ $j^*=j_*+1$).{}\ {}\hfill\slutt
\end{itemize}
Whenever necessary to distinguish the ex ante (at $j_0$) random
variable which is either 0 (if $D_j$ does not occur) or $>0$ on one
hand, from the $D_j$-conditional positive return on the other --
colloquially speaking, the lottery ticket that yields $\eqslantgtr0$
from the actual positive lunch prize -- we shall use terms like the
``event'' that the lunch ``manifests itself''.
\end{defn}

Obviously, the lack of time bound makes no difference for an
arbitrage; if there is an arbitrage according to this definition, then
for some fixed $Q$, there is an arbitrage which is closed out within
$Q$ steps, i.e., $j^*\leq Q$. Conversely, it does not matter that
$\period$ is assumed deterministic; had we employed the same
definition except with $\period$ being merely measureable and finite,
we would have had an arbitrage for some deterministic $\period$ as
well.  The FLVR definition, on the other hand, might require an
unbounded $j_*$. Informally, a FLVR is a sequence of lunches with
uniformly positive mean, but where the risk tends to zero. This means
that for any nonzero downside you choose as tolerance, then there is a
fixed $Q$ such that you have a lunch within your risk tolerance within
$Q$ steps. Letting downside tend to zero, then our setup allows $Q$ to
grow, as long as $\period$ obeys a fixed bound; compare this to the
usual Black--Scholes setup, which rules out the strategy of trading
and waiting for the unbounded stopping time until your position has
made a given profit.

For the purpose of giving \emph{sufficient\/} conditions for
arbitrage/FLVR under \emph{small\/} transaction cost -- which is the
main object of this section -- the simple model \eqref{profit}, for
which $\period=1$, turns out fairly close to general:

\begin{prop}[Free lunches in the simple model \eqref{profit} vs.\ in
  the full model]\label{simplevsfull}
  Fix $\units>0,\ n<\infty$. Assume that $\Lambda^*$ of at most linear
  growth wrt.\ the last variable (the selling price). Then there is an
  arbitrage for sufficiently small $\lambda$, provided that so is the
  case in the simple model \eqref{profit}. If $G$ is convex, then
  there is FLVR for sufficiently small $\lambda$, provided that so is
  the case in the simple model \eqref{profit}.
\end{prop}

\begin{proof}
  The proof is less interesting, and is relegated to the
  Appendix. Notice that if $\xi_jK(\sn {j+1},\sn j)$ is upper bounded
  (for every $j$ and $n$), then the at most linear growth condition
  will hold (since an arbitrage must be closed out in a bounded number
  of periods).
\end{proof}

Informally, an arbitrage in the simple model \eqref{profit}, occurs if
at some {bounded} ${j_*\in [j_0,\infty)}$, given the information
available then, the {transaction costs plus the worst-case possible
}downside from the innovation $\zzz^{(n)}_{{j_*}+1}$ will be more than
fully compensated by the contribution from the dependence of the past
(i.e. $\xxx^{(n)}_{{j_*}+1}+\yyy^{(n)}_{{j_*}+1}$); a FLVR occurs if
it is ``sufficiently more than compensated in mean'' and ``nearly
fully compensated in worst-case''. The following result is key for the
arbitrage case:

\begin{prop}[Sufficient conditions for arbitrage in the simple model
  \eqref{profit}] Fix $n<\infty$.  If for some natural $j\geq j_0$ we
  have
  \begin{align}
    \esssup_{\{\xi_i\}_{i=j_0+1,\dots,j}}\{\xxx_{j}+\yyy_{j}\}
    +\essinf_{\xi_{j+1}}\zzz_{j+1}%
    \geq\bar\lambda\geq0\label{x+y+z>lambda}
\end{align}
we have an arbitrage for all transaction costs
$\lambda\in[0,\bar\lambda)$ by choosing $j_*=$ this $j$. Furthermore,
we have arbitrage for transaction cost $\bar\lambda$ if in addition
there is a point probability that $\zzz_{j+1}$ attains its
$\esssup{}$.
\end{prop}
\begin{proof}
  Suppose \eqref{x+y+z>lambda} holds for some $\bar\lambda$. Let $D_j$
  be the $j$-measurable event of attaining
$$\xxx_{j}+\yyy_{j}\geq
\esssup_{\{\xi_i\}_{i=j_0+1,\dots,j}}\{\xxx_{j}+\yyy_{j}\}-\epsilon.$$
Then $\Pr_{j_0}[D_j]>0$ for each $\epsilon>0$; let
$\epsilon\in(0,\bar\lambda-\lambda)$ if nonempty. Should the event
$D_j$ occur at step $j$, then \eqref{profit_a} is $>0$ and so is the
$\essinf$ of \eqref{ratio*}, where then both numerator and denominator
become positive. Arbitrage also for transaction cost $\bar\lambda$
holds if we have positive probability at $\epsilon=0$.
\end{proof}

So the discrete market will admit an arbitrage if there may occur a
period so good that the contribution from this beneficial history,
knocks out the innovation so much that the worst-case scenario is a
profit. Evidently, this will not happen if $K$ is a constant (i.e.\
ordinary Brownian motion), for then the history does not matter; on
the other hand, if $K$ is increasing in its first variable, then it is
near-trivial to construct arbitrage examples by letting downside be
bounded and the upside be unbounded.  While one can certainly imagine
a some modeler trying to use e.g.\ a shifted lognormal in order to
model limited liability investments, a choice of a symmetric
distribution for the $\xi_i$ would arguably be more natural and
innocuous-looking. But for a suitably wide range of models, a larger
downside than upside will not even prevent an arbitrage, as we shall
soon see. For the book-keeping of ``good'' and ``bad'' $\xi_i$
outcomes, we shall denote their essential suprema/infima as follows:
\begin{subequations}
\begin{align}
  \label{eq:Mmgb}
  M_i=M_i^{(n)}&=\esssup \xi_i^{(n)}\\
  m_i=m_i^{(n)}&=-\essinf \xi_i^{(n)},%
\end{align}
\end{subequations}

{Standing at step $j$, then the worst that can
  happen in the next innovation, cf.\ \eqref{zzz}, is denoted
  $\beta_j$ (``beta'' for ``bad''):} %
\begin{subequations}
\begin{align}
  \beta_j=\beta_j^{(n)}&=
\begin{cases}
  -m_j&\quad\text{if }K(\sn{j+1},\sn j)\geq0\\
  M_j&\quad\text{otherwise.}
\end{cases}
\intertext{Looking back in time, we define $\gamma_{ij}$ (``gamma''
  for ``good'') to be the best possible history over $i=j_0,\cdots,j$
  (cf.\ \eqref{yyy}):} \gamma_{ij}=\gamma_{ij}^{(n)}&=
\begin{cases}
  M_i&\quad\text{if }K(\sn{j+1},\sn i)\geq K(\sn j,\sn i)\\
  -m_i&\quad\text{otherwise.}
\end{cases}
\end{align}
\end{subequations}
We want to specify this in terms of time, not only steps. Suppose we
are targeting an arbitrage within time $T$ for the discretised
model, choosing $\jT$ to be the second-to-last step before time
$T$, closing out the transaction before time $T$:
\begin{align}
  \label{eq:jj}
  \jT+1=\big\lfloor nT-nt_0+\lfloor nt_0\rfloor\big\rfloor
\end{align}
Then define $\Gamma(T,s)=\Gamma^{(n)}(T,s)$ as a left-continuous step
function with values $\Gamma(T,\sn i)=\gamma_{i\jT}$; extend it to be
$0$ for $s\in[\sn {\jT},T]$. Consider then $n^{-1/2} y_{\jT}$ and bear
in mind that $\jT$ depends on $n$ chosen as \eqref{eq:jj}. Then,
provided that limits exist (again, $K'_1$ denotes the derivative wrt.\
the first variable), we have
\begin{align*}
   \lim_n (\esssup y_{\jT} \sqrt n) \geq 
  \int_{t_0}^{T}K'_1(T,s)\liminf_n\Gamma^{(n)}(T,s)\d s
\end{align*}
where the inequality follows from nonnegativity of the integrand and
the Fatou lemma. So, given $\epsilon>0$, then for all large enough but
finite $n$, there will be positive $\Pr_{j_0}$ measure of the event
\begin{align}\label{integralkriterium}
  D_{j_*}=\Big\{y_{j_*} \geq
  \Big[\int_{t_0}^{T}K'_1(T,s)\Gamma^{(n)}(T,s)\d s-\epsilon\Big]n^{-1/2}\Big\}.
\end{align}
This gives rise to the following result:

\begin{thm}[Sufficient conditions for arbitrage within time $T$ in the
  simple model \eqref{profit}]\label{thm:arbitrage}
  Fix a $T>t_0$ and for each $n$, let $\jT$ be given by
  \eqref{eq:jj}. Assume that at $T$ we have $\JODD$ H\"o{}lder continuous
  with exponent $\alpha>1/2$, and furthermore that $K(t,s)$ is
  differentiable in the first variable, at $t=T$, for each
  $s\in(t_0,T)$.  Then if
  \begin{align}
    \label{eq:3}
    \int_{t_0}^{T} K'_1(T,s)\Gamma^{(n)}(T,s)\d s>\big\lvert K(\sn{\jT+1},\sn{\jT})\beta_j^{(n)}\big\rvert+\bar\epsilon
  \end{align}
  holds for all large enough $n$, some $\bar\epsilon>0$, then for any
  $n$ large enough, there is an arbitrage with sufficiently small
  transaction costs, by waiting until step $j_*$.
\end{thm}

\begin{proof}
  Under the assumptions, we would on the event in
  \eqref{integralkriterium} have a net return of at least
  \begin{align*}
    \bigg[\frac{a(\sn {\jT})+ \JODD(\sn {\jT+1})-\JODD(\sn
      {\jT})}{n^{1/2}}&+\int_{t_0}^{T}K'_1(T,s)\Gamma^{(n)}(T,s)\d s
    \\&-\lvert
    K(\sn{\jT+1},\sn{\jT})\beta_j^{(n)}\big\rvert-\epsilon\bigg]n^{-1/2}-\lambda
  \end{align*}
  and the H\"o{}lder regularity ensures that the first term inside the
  bracket (i.e.\ $x_j\sqrt n$), will vanish as $n$ grows. Then by
  \eqref{eq:3}, the net return will for large enough $n$ be a positive
  random variable, even with small transaction costs.
\end{proof}

\begin{rem}\label{KTT=0}
  First, observe that if $\lim_{s\nearrow T}K(T,s)=0$, then this would
  lead to arbitrages.  Second, note that Theorem~\ref{thm:arbitrage}
  is stated for fixed $T$, but it is sufficient to look for some $T$
  where it applies.  For example, if $\xi_i$ have symmetric support
  for each $n$, then we can replace \eqref{eq:3} by
  \begin{align}
    \label{eq:4}
    \sup \Big\{\int_{t_0}^{T} \big\lvert{K'_1(T,s)}\big\rvert\d s
    -\big\lvert K(T,T-\frac1n)\big\rvert\Big\}>0
  \end{align}
  where the $\sup$ is taken over those $T>t_0$ for which $n(T-t_0)$ is
  integer.  Now one can look for arbitrages by letting $T$ grow.\slutt
\end{rem}

Theorem \ref{thm:arbitrage} also applies to semimartingales. The
corollary is stated only for the natural choice of symmetric
innovations:

\begin{cor}\label{semimg}
  There are infinite-variation semimartingales $Z$, {equalling weak
    limits of their discretisations} $Z^{(n)}$ formed by i.i.d.\
  bounded symmetric $\xi_i$, for which Theorem
  \ref{thm:arbitrage} applies.
\end{cor}

\begin{proof}
  Put $t_0=0$ for simplicity. From \citet[Theorem 3.9]{MR2024843}, it
  is sufficient for the semimartingale property that
  $K(t,s)=\ELL(t-s)$ on $t>s>0$ with $\ELL$ being continuous and
  piecewise differentiable with $\ELL'\in\kl L^2((0,\infty))$, and
  under these conditions, total variation is infinite on compacts iff
  $\ELL(0^+)\neq0$. Choose a $\ELL\geq0$ with a global maximum at $T$,
  with $\ELL(T)>2\ELL(0^+)>0$; then it satisfies the hypothesis of
  Theorem~\ref{thm:arbitrage}, and we only need $\ELL(\vartheta)$ to
  be smooth and $\kappa'$ to tend sufficiently fast to $0$ as to be
  square integrable.
\end{proof}

The form where the dependence on $(t,s)$ only appear through the
difference, will cover many cases and simplify calculations. We
introduce the conditions:
\begin{subequations}\label{L+}
  \begin{gather}
    \label{samesupp}\text{For each $n$, we have $i$-independent $m_i=m=m^{(n)}$ and $M_i=M=M^{(n)}$.}
    \\
\label{eq:L}
K(t,s)=\ELL((t-s)_+)\qquad\text{for $s\geq t_0$, with $\ELL$ not constant on
  $(0,\infty)$}
\end{gather}
\end{subequations}
-- the non-constantness ruling out the ordinary Brownian motion.  As
seen above, this form covers a wide class of even semimartingales. We
can then write $\yyy$ and $\zzz$ as
\begin{subequations}
\begin{align}
  \yyy_{j} &= \frac1{\sqrt n}
  \sum_{i=j_0}^{j-1}\Big[\ELL(\tfrac{j+1-i}{n})-\ELL(\tfrac{j-i}{n})\Big]\:\xi^{(n)}_{i+1} \label{yyy_L}\\
  \zzz_{j+1}&= \frac1 {\sqrt n} \: \ELL(\tfrac1n)\:\xi^{(n)}_{j+1}  \label{zzz_L}
\end{align}
\end{subequations}
Now consider the good outcomes $\gamma_{ij}$; if $\ELL$ is monotone or
$m=M$, then the series in~\eqref{yyy_L} will telescope. A non-monotone
$\ELL$ only has more variation, which increases the sum, so the
ess~sup of $\yyy_j$ will therefore be at least
\begin{align}\left\vert\frac{\ELL(\tfrac{j-j_0+1}{n})-\ELL(\tfrac1n)}{\sqrt n}\right\vert\cdot
\begin{cases}
  M^{(n)} &\qquad\text{if }\ELL(\tfrac{j-j_0+1}{n})\geq \ELL(\tfrac1n)\\m^{(n)}&\qquad\text{otherwise.}
\end{cases}
\end{align}
(if we want to utilise the variation of $\ELL$, we could write in
terms as a sum of $|\ELL'|$-terms, tending to the constant times the
total variation of $\ELL$ over the interval $(1/n,(j-j_0+1)/n)$.) We
have the following:

\begin{thm}[Sufficient conditions for arbitrage in the simple model
  under the form \eqref{L+}]\label{thm:L} Suppose that $\JODD$ is
  H\"o{}lder
  continuous with index $\alpha>1/2$ and furthermore that for some
  subsequence $n_\ell$, we have ${m^{(n_\ell)}}$, and ${M^{(n_\ell)}}$
  bounded. Then, each of the following conditions implies arbitrage
  for all large enough $\ell$ -- and furthermore, for each of those
  $n_\ell$, the arbitrage admits small enough transaction costs:
  \begin{enumerate}%
  \item $\liminf_\ell|\ELL(1/n_\ell)|=0$, or $\ELL$ changes
    sign. \label{nullellerfortegnsskift}
  \item The total variation of $\ELL$ over $(0,\infty)$ (i.e.\
    $\int_{0}^\infty \left| \ELL'(\vartheta)\right|\d\vartheta$ if
    $\ELL'$ exists), is $>$ than
    $$\liminf_\ell\Big[\big|\ELL(1/n_\ell)\big|
    \cdot\frac{\max\{M^{(n_\ell)},m^{(n_\ell)}\}}{\min\{M^{(n_\ell)},m^{(n_\ell)}\}}\Big].$$
  \end{enumerate}
\end{thm}
\begin{proof}{\ }In all cases, H\"o{}lder regularity ensures that
  $x_j\sqrt n $ will tend to $0$ as $n$ grows. Then:
  \begin{enumerate}%
  \item By \eqref{eq:L}, $\ELL$ takes some nonzero value. Suppose
    first that $\ELL(1/n_\ell)\searrow0$ while
    $\ELL(\vartheta)>0$. Then choosing a sequence of $j$'s so that
    $(j-j_0+1)/n$ approximates $\vartheta$ from the appropriate side
    (recall that \eqref{Ka} assumes only piecewise continuity), we
    obtain
    \label{proofa}
    \begin{align*}
      \sqrt n \:\cdot\Big[z_{j+1}+\esssup y_{j+1}\Big]%
      &\geq M\cdot \Big[\ELL(\tfrac{j-j_0+1}{n})-
      \ELL(\tfrac1n)-\tfrac mM \ELL(\tfrac1n)\Big]%
      \\&\rightarrow M\ELL(\vartheta)>0%
    \end{align*}
    The negative-sign case likewise converges to $m\cdot|\ELL(\vartheta)|$.
    \\
    Now suppose that $\ELL$ changes sign, and by the previous part of
    the proof, we can assume that $\ELL(1/n_\ell)$ is bounded away
    from $0$. Suppose that $\ELL(1/n_\ell)>0>\ELL(\vartheta)$. Then
    choosing $j$ as above, we obtain
    \begin{align*}
      \sqrt n \:\cdot\Big[z_{j+1}+\esssup y_{j+1}\Big]%
      & \geq m\cdot \Big[ \ELL(\tfrac1n)-\ELL(\tfrac{j-j_0+1}{n})-
      \ELL(\tfrac1n)\Big]%
      \\&=-m \ELL(\tfrac{j-j_0+1}{n})
    \end{align*}
    which is positive whenever $n_\ell$ is large enough. The case with
    reversed signs follows likewise.

  \item $\sqrt n\ \esssup \yyy_j$ exceeds $\min\{m^{(n)},M^{(n)}\}\
    \times$ total variation on $(\tfrac1n,\tfrac{j-j_0+1}{n})$, while
    even in the worst case, $\sqrt n\ z_j\geq - |\ELL(1/n)|\cdot
    \max\{m^{(n)},M^{(n)}\}$.
  \end{enumerate}
  \end{proof}
  The H\"o{}lder regularity condition on $J$ in
  Theorems~\ref{thm:arbitrage} and \ref{thm:L} admits ramifications,
  as we need only bound the downside -- it can be replaced by the
  condition that for each $n$ we have $x_j\geq0$ for infinitely many
  $j$, which by symmetry of $W$ occurs in at least half of the cases
  (unconditionally, i.e.\ ``$\Pr_{-\infty}$''). Should $x_j\,\sqrt n$
  blow up as $n$ grows, then it would be expected that
  $\JODD(T,T-1/n)$ oscillates around $0$, and we would be able to
  extract a subsequence where it adds positively to the return.
  Similar considerations would improve upon the next
  Theorem~\ref{thm:FLVR} as well.  Before we state that result, it
  should be noted that the total variation criterion can also be
  improved upon in the setup of Theorem~\ref{thm:arbitrage}, if the
  variation of the step function corresponding to grid size $1/n$,
  diverges as time grows. When $|\ELL(1/n)|\to\infty$, and $\ELL$ is
  monotone (anything else improves total variation) and does not
  change sign (if it does, Theorem~\ref{thm:arbitrage}
  part~\ref{nullellerfortegnsskift} applies), we can still have free
  lunches with vanishing risk:

  \begin{thm}[Sufficient conditions for FLVR under the form
    \eqref{L+}]\label{thm:FLVR}
    Suppose zero transaction cost and that $M^{(n)}=m^{(n)}$
    and \begin{math} \inf_i \E[\xi_i^{(n)}] / m^{(n)}>-1.
    \end{math} Furthermore, assume that $\ELL$ does not change sign,
    and that $\inf_\vartheta |\ELL(\vartheta)|=0<\liminf_n
    |\ELL(1/n)|$ (possibly $=+\infty)$). Then either of the following
    is sufficient for FLVR:
  \begin{enumerate}
  \item For given $n$: \label{FLVRa} $x_j^{(n)}\geq0$ for infinitely
    many $j$, and additionally,
    $\lim_{\vartheta\rightarrow\infty}\ELL(\vartheta)=0$.
  \item \label{FLVRb} $\JODD$ is H\"o{}lder continuous with index
    $>1/2$, \begin{math}
      \liminf_n %
      \frac{\inf_i \E[\xi_i^{(n)}]}{ m^{(n)}}>-1
    \end{math}, and $n$ is large enough.  The FLVR is an arbitrage if the
    positivity of the $x_j$ is uniform (wrt.\ $j$).
  \end{enumerate}
\end{thm}
\begin{proof} 
  We prove only the case of positive $\ELL$. Just like in the proof of
  Theorem~\ref{thm:L} part~\ref{nullellerfortegnsskift}, $\sqrt n\
  \big(\esssup y_j+\essinf z_{j+1} \big)$ will telescope to $-m^{(n)}
  \ELL(\tfrac{j-j_0+1}{n})$, which can be made arbitrarily close to
  $0$; a $\Pr_{j_0}$-positive event will be it falling within
  $\epsilon/n$ of $-m^{(n)} \ELL(\tfrac{j-j_0+1}{n})$. It will turn
  out that this takes care of the numerator of \eqref{ratio*} in the
  FLVR definition. For the denominator, we need the expected return:
  \begin{align}
    x_j+\esssup y_j+\E[z_{j+1}]&=x_j+\frac1 {\sqrt n}\Big[-m
    \ELL(\tfrac{j-j_0+1}{n})+m \ELL(\tfrac1n)+
    \ELL(\tfrac1n)\E\xi_{j+1}\Big] %
    \notag\\ &%
    = x_j+ \frac m {\sqrt
      n}\Big[- \ELL(\tfrac{j-j_0+1}{n})+ \ELL(\tfrac1n)\big(1+
    \frac{\E\xi_{j+1}}m\big)\Big]\label{FLVRbracket}
    \end{align}
    \begin{enumerate}
    \item %
      Now passing through a subsequence with nonnegative $x_j$, then
      $\ELL(\tfrac{j-j_0+1}{n})$ will vanish and \eqref{FLVRbracket}
      will be positive from the assumption; this takes care of the
      denominator of~\eqref{ratio*}. Assuming that the numerator is
      negative (otherwise there is arbitrage), the
      ratio~\eqref{ratio*} will on the $\Pr_{j_0}$-positive event of
      the history falling within $\epsilon/\sqrt n$ of its $\esssup$,
      exceed
      \begin{align*}
        \frac{-m \ELL(\tfrac{j-j_0+1}{n})-\epsilon}
        {m \Big[- \ELL(\tfrac{j-j_0+1}{n})+ \ELL(\tfrac1n)\big(1+
          \frac{\E\xi_{j+1}}m\big)\Big]}
      \end{align*}
      which tends to $0$ as $j$ and $\epsilon^{-1}$ grow.
    \item The ratio~\eqref{ratio*} becomes -- with
      $\epsilon=\epsilon_n$ possibly $n$-dependent --
      \begin{align*}
        \frac{%
          \frac{m}{\sqrt n}%
          \Big[\frac{x_j\sqrt n}{m}-\ELL(\tfrac{j-j_0+1}{n})
          -\frac{\epsilon}{m\sqrt n}\Big]%
        }%
        { \frac{m}{\sqrt n}%
          \Big[\frac{x_j\sqrt n}{m}-\ELL(\tfrac{j-j_0+1}{n})
          -\frac{\epsilon}{m\sqrt n}+ \ELL(\tfrac1n)\big(1+
          \frac{\E\xi_{j+1}}m\big)\Big]%
        }
      \end{align*}
      First, cancel $m/\sqrt n$. Then, observe that as in the proof of
      Theorem~\ref{thm:L} part~\ref{proofa}, we can choose $j$
      depending on $n$ as to approximate the appropriate $\vartheta$,
      or possibly the appropriate sequence of $\vartheta$'s, so that
      $\ELL(\tfrac{j-j_0+1}{n})$ vanishes in the limit, along with --
      by assumption -- everything involving $x_j$. Then for suitably
      small $\epsilon$, then $\ELL(1/n)$ will make the denominator
      (and hence the expectation) positive, while the numerator can be
      chosen arbitrarily small.
    \end{enumerate}
  \end{proof}

\begin{rem}\label{rem:stygghistorikk}
  Notice that the statements of Theorem~\ref{thm:L} and of
  Theorem~\ref{thm:FLVR} part~\ref{FLVRb}, do not depend on what $t_0$
  is, and what history the agent faces upon entry. It is certainly not
  obvious that it should be this way.  The setup of
  Theorem~\ref{thm:arbitrage} does not rule out a priori that there
  could be an arbitrage initially, to be exploited at a later stage if
  a positive event $D_j$ occurs, but which with positive probability
  disappears for good. (In more formal terms, the $\xi_{j+1}$ could be
  drawn so that not only would $\Pr_{j_0+1}[D_j]=0$, but also
  $\Pr_{j_0+1}[\Pr_i[D_j]>0]=0$ for all $i>j_0$.) But under the
  applicability of Theorem~\ref{thm:L} or Theorem~\ref{thm:FLVR}
  part~\ref{FLVRb}, this will not be the case: the arbitrage, resp.\
  FLVR, will show up for large enough $n$ regardless of whether the
  agent enters after a long ``bad'' period which hampers future
  prospects; for fine enough discretisation, there will always be a
  positive event where the arbitrage could materialise. This is not to
  say that the probability of this event is independent of history,
  nor that the choice of positive event is -- only the binary question
  of existence.  Under these results, regardless how disadvantageous
  the development has been, the strategy of calmly waiting for lunch
  time, will always yield positive expected value.\slutt
\end{rem}

\section{Some examples and non-examples}
We will in the following discuss a few cases. Throughout this section,
assume common symmetric support $[-m,m]$.
\begin{enumerate}%
\item Brownian motion is not prone to arbitrages in the discretised
  version; we have $K(t,s)=1_{t>s>0}$, so there is no contribution
  from history.
\item The Ornstein--Uhlenbeck process (mean-reverting to $0$) admits
  the representation $\ELL(\vartheta)=\ELL(0)\cdot e^{-v t}$ with
  $v>0$. This $\ELL$ satisfies Theorem~\ref{thm:FLVR}, which will
  yield FLVR (but, easily, not arbitrage) if choosing the distribution
  of the $\xi_i$ to comply with the assumptions. If there is positive
  drift (mean-reversion to a positive level $\mu$), then the
  discretised version admits arbitrage.  It should be noted though,
  that in the continuous-time model, a portfolio of $\eta(t)$ yields a
  wealth process dynamics of %
  $\eta(t)\d Z(t)=\eta(t)[v\cdot(\mu-Z(t))+\sigma\d W(t)\big]$, where
  there is an arbitrarily big upside for given volatility level, by
  waiting for $Z$ to become negatively large. However, the continuous
  model remains arbitrage-free, regardless of $\mu$.
\item \authorrefnametobedeleted{Sottinen }\citet{MR1849425} considers
  fractional Brownian motion with Hurst parameter $H>1/2$, using the
  representation $$K(t,s)=\int_s^t(u/s)^{H-1/2}(u-s)^{H-3/2}\d u$$ (up
  to an irrelevant positive constant), so that $K(t^+,t)=0$ and $K'_1$
  is positive. Then by Remark~\ref{KTT=0} we will have Theorem
  \ref{thm:arbitrage} applying, as
  $\JODD(t)=\int_{0}^{t_0}\int_s^t(u/s)^{H-1/2}(u-s)^{H-3/2}\d u\, \d
  W(s)$ is differentiable at $t=T^-$ (just interchange order of
  integration).

  Furthermore, by Remark~\ref{rem:stygghistorikk}, the arbitrage holds
  regardless of history. No matter how bad (and how long!) the initial
  period until entry is, there is still a positive event that a free
  lunch will actually manifest.
  \item Maybe a more common representation for fractional Brownian
motion is, for any $H\neq1/2$ and up to a constant, 
$$K(t,s)=(t-s)^{H-1/2}-((-s)_+)^{H-1/2},$$
corresponding to
$\ELL(\vartheta)=\vartheta^{H-1/2}$.  Let us assume that the $\xi_i$
have the same support.  $\ELL$ is monotonous, so then conditions
\eqref{L+} hold.  Now the results are different for positively
($H>1/2$) and negatively ($H<1/2$) autocorrelated fBm:
\begin{itemize}\item In the case $H>1/2$, $\ELL(0^+)=0$ and $\ELL$ is
  increasing.  Furthermore, $\JODD$ is H\"o{}lder continuous of order
  up to $H$. Theorem \ref{thm:arbitrage} applies, by
  Remark~\ref{KTT=0}.
\item In the case $H<1/2$, $\ELL(0^+)=\infty$ while
  $\ELL(\infty)=0$. Then for at least half of the cases,
  Theorem~\ref{thm:FLVR} part~\ref{FLVRa} applies.
\end{itemize}
\label{item:fBm}
\item \citet{MR1434408} proposes a modification of fractional Brownian
  motion, in order to eliminate the arbitrage but preserve the long
  run memory properties which motivated the use of fBm in finance in
  the first place. Rogers gives a specific (monotone) example
  \begin{align}\label{Rogers}
    \ELL(\vartheta)=k\cdot (\vartheta^2+v)^{(H-\frac12)/2},
  \end{align}
  but suggests more generally to choose $\ELL$ such that $\ELL(0)=1$,
  $\ELL'(0)=1$, and has the same $\sim \vartheta^{H-1/2}$ behaviour
  for large $\vartheta$.  This behaviour, tending to $\infty$ for
  $H>1/2$ and $0$ for $H<1/2$, is sufficient to yield the same results
  as for example~\ref{item:fBm}.\label{item:Rogers}
\item \label{item:mix} For a mix between a fractional and an
  (uncorrelated) ordinary Brownian motion, there is a very peculiar
  result by \authorrefnametobedeleted{Cheridito }\citet{MR1873835}: if
  the fBm part has $H>3/4$, then it behaves just as drift (which for
  example means that it does not enter in the Black--Scholes
  formula). For $H\leq3/4$, there is still arbitrage as if there were
  no Brownian component. However, in our case, such a process works
  like example~\ref{item:Rogers} when $H>1/2$: mixing in Brownian
  motion at volatility $\sigma$, we get $\ELL(t-s)$ replaced by
  $\sqrt{\sigma^2+\ELL(t-s)^2}=(\sigma^2 + (t-s)^{2H-1})^{1/2}$ which
  for $H>1/2$ works ana\-logous to \eqref{Rogers}, and will admit
  arbitrage in the discretisation. There is apparently nothing
  happening at the Cheridito threshold of $3/4$.
\end{enumerate}

We end this section by pointing out that not only are the discrete
markets possibly different than their weak limits when it comes to
\emph{existence\/} of free lunches; the arbitrages themselves might
occur from different properties of $K$.  In the canonical models,
either of the properties $\ELL(0^+)=0$ and $\ELL(+\infty)=+\infty$
will lead to arbitrages, and the latter will be due to the long-run
memory. The long-run memory was arguably the reason why fractional
Brownian motion $\ELL(\vartheta)=\vartheta^{H-1/2}$ (with $H>1/2$) was
suggested in the first place as driving noise for financial markets,
and the Rogers proposal of example~\ref{item:Rogers} above, leaves
that long memory in the process. Let assume that $t_0=0$, and that,
coincidentally, $x_{j_0}=x_{j_0+1}=0$ in order to isolate short-term
effects.  Fix $n$ for the moment. Then the first-step innovation is
symmetric, and the next one cannot lead to arbitrage either, as
$\ELL(2/n)-\ELL(1/n)<\ELL(1/n)$ by concavity. The minimum number of
steps (after $j_0$) for the arbitrage for the fBm case, is the
smallest integer $> 2^{1/(H-1/2)}$ (equality suffices if the $\esssup$
has point mass), so that the absolute minimum is 5 steps, obtained for
$H$ above $\approx 0.931$. Of course, as the partition refines and $n$
grows, this happens in shorter \emph{time}, thus approaching the
continuous-time setup where the profit instantly increases from $0$
(\authorrefnametobedeleted{this author, }\citet{NCF-FL}). On the other
hand, even when mixed with ordinary Brownian motion, the
examples~\ref{item:Rogers} and~\ref{item:mix} yield arbitrages;
boosting up the ordinary Brownian part, will merely require a longer
beneficial period before the arbitrage manifests itself. Those
arbitrages are due to the \emph{long run\/} behaviour -- namely, the
fact that $\ELL$ tends to $+\infty$ on the long run.

\section{Concluding remarks}
We have seen that discrete-time random walk markets may behave
radically different from their weak limits, as the former may admit
arbitrages or FLVRs which vanish in the limit. Furthermore, quite a
few of our estimates may be sharpened and the results likely ramified.
That is in the author's opinion not a main concern. Rather, it has
turned out that a type of result once interpreted as another
objectional property of the fBm's, is simply to be expected if one
models moving average processes in such a careless way.

One could certainly try to remedy the problem by choosing wide
supports with low probability of the arbitrages manifesting
themselves. Arguably, a practitioner should be worried even at far
less radical modeling issues than the binary \emph{existence of\/}
arbitrage, and a quick fix which merely covers up the most obvious
undesirable property, might not be an adequate solution to the
inherent problem.

\appendix

\newcommand{\zetabuy}{\zeta_*}
\newcommand{\zetasell}{\zeta^*}
\newcommand{\Sbuy}{S_*}
\newcommand{\Ssell}{S^*}

\section*{Appendix: Proof of Proposition~\ref{simplevsfull}}

Put $\units=\period=1$. Denote the buying and selling prices in the
simple model \eqref{profit} by $\zetabuy$ and $\zetasell$, and in the
full model by $\Sbuy=G(\zetabuy)$ and $\Ssell=G(\zetasell)$. Observe
first that we may take $\zetabuy$ bounded by restricting $D_j$ without
avoiding the property $\Pr_{j_0}[D_j]>0$, and we will do so in the
following. Assume that the simple model \eqref{profit} has arbitrage
for transaction cost $c>0$; then
  \begin{align}\zetasell&\eqslantgtr\zetabuy+c\notag
    \intertext{which by applying $G$ and rearranging, is equivalent to}
    \Ssell-\Sbuy&-\lambda(\Lambda^*(1,\Sbuy,\Ssell)+\Lambda_*(1,\Sbuy))%
    \notag\\&%
    \eqslantgtr G(\zetabuy +c)-G(\zetabuy)-\lambda(\Lambda^*(1,G(\zetabuy),G(\zetasell))+\Lambda_*(1,G(\zetabuy)))\notag
    \intertext{so we have arbitrage if the right hand side is nonnegative, so it suffices that}0<\lambda&\leq\frac{G(\zetabuy +c)-G(\zetabuy)}{\Lambda^*(1,G(\zetabuy),G(\zetasell))+\Lambda_*(1,G(\zetabuy))}
\label{graveacute}
  \end{align}
  As already remarked, we can assume $\zetabuy$ bounded, so it
  suffices to bound $\Lambda^*$ for given $\zetabuy$. By at most
  linear growth, 
    \begin{math}
      \Lambda^*(1,\Sbuy,\Ssell)\leq\lambda_0(\Sbuy) +\lambda_1(\Sbuy)\cdot \Ssell,
    \end{math}
    where $\lambda_0$, $\lambda_1$ are locally bounded functions of
    $\Sbuy$, the return is%
\begin{align*}
    Y&=\Ssell-\Sbuy-\lambda(\Lambda^*(1,\Sbuy,\Ssell)+\Lambda_*(1,\Sbuy))
    \\&\geq\Ssell(1-\lambda\lambda_1(\Sbuy))-\Sbuy(1-\lambda\lambda_1(\Sbuy))-\lambda\lambda_1(\Sbuy)\Sbuy -\lambda(\Lambda_*+\lambda_0)
    \intertext{By boundedness of $\Sbuy$ we can take
      $\lambda\lambda_1(\Sbuy)<1$, in which case the return will be
      $\eqslantgtr0$ if} \Ssell-\Sbuy &\eqslantgtr \lambda
    \frac{\lambda_0+\lambda_1\Sbuy
      +\Lambda_*}{1-\lambda\lambda_1(\Sbuy)}
  \end{align*}
  which at least equals $\lambda\cdot 2(\lambda_0+\lambda_1\Sbuy
  +\Lambda_*)=:\lambda\tilde\Lambda$, where $\tilde\Lambda$ is a
  locally bounded function of $\Sbuy$ alone. Hence we can consider
  the problem with $\tilde\Lambda$ in place of $\Lambda_*$ and
  assuming $\Lambda^*=0$, and then the right-hand side of
  \eqref{graveacute} will not depend on the selling price
  $\zetasell$. We are done with the arbitrage part of the
  proposition.
  
  For FLVR, it suffices to point out that Jensen's inequality improves
  both the upside and downside for convex $G$, compared to for linear
  ones.

\bibliographystyle{apalike}
\bibliography{1206.5756}

\begin{thebibliography}{}

\bibitem[Cheridito, 2001]{MR1873835}
Cheridito, P. (2001).
\newblock Mixed fractional {B}rownian motion.
\newblock {\em Bernoulli}, 7(6):913--934.

\bibitem[Cheridito, 2004]{MR2024843}
Cheridito, P. (2004).
\newblock Gaussian moving averages, semimartingales and option pricing.
\newblock {\em Stochastic Process. Appl.}, 109(1):47--68.

\bibitem[Framstad, 2004]{NCF-FL}
Framstad, N.~C. (2004).
\newblock Arbitrage with the ordinary chain rule.
\newblock {\em Finance Letters}, 2(6).

\bibitem[Guasoni, 2006]{MR2239592}
Guasoni, P. (2006).
\newblock No arbitrage under transaction costs, with fractional {B}rownian
  motion and beyond.
\newblock {\em Math. Finance}, 16(3):569--582.

\bibitem[Guasoni et~al., 2008]{MR2398764}
Guasoni, P., R{\'a}sonyi, M., and Schachermayer, W. (2008).
\newblock Consistent price systems and face-lifting pricing under transaction
  costs.
\newblock {\em Ann. Appl. Probab.}, 18(2):491--520.

\bibitem[Rogers, 1997]{MR1434408}
Rogers, L. C.~G. (1997).
\newblock Arbitrage with fractional {B}rownian motion.
\newblock {\em Math. Finance}, 7(1):95--105.

\bibitem[Shiryaev, 1999]{MR1695318}
Shiryaev, A.~N. (1999).
\newblock {\em Essentials of stochastic finance}, volume~3 of {\em Advanced
  Series on Statistical Science \& Applied Probability}.
\newblock World Scientific Publishing Co. Inc., River Edge, NJ.
\newblock Facts, models, theory, Translated from the Russian manuscript by N.
  Kruzhilin.

\bibitem[Sottinen, 2001]{MR1849425}
Sottinen, T. (2001).
\newblock Fractional {B}rownian motion, random walks and binary market models.
\newblock {\em Finance Stoch.}, 5(3):343--355.

\bibitem[Whitt, 2007]{MR2368952}
Whitt, W. (2007).
\newblock Proofs of the martingale {FCLT}.
\newblock {\em Probab. Surv.}, 4:268--302.

\end{thebibliography}

\end{document}